\documentclass[conference]{IEEEtran}
\usepackage[numbers,sort&compress]{natbib}
\usepackage{amsmath,amssymb,amsfonts}
\usepackage{algorithmic}
\usepackage{graphicx}
\usepackage{textcomp}
\usepackage{xcolor}
\usepackage[capitalize]{cleveref}
\usepackage{multirow}
\usepackage[shortlabels]{enumitem}
\usepackage{dblfloatfix}    % To enable figures at the bottom of page

\usepackage{siunitx}
\usepackage{import}

\usepackage[printonlyused,nohyperlinks,nolist]{acronym}

\usepackage{amsthm}
\usepackage{dsfont}
\usepackage{amssymb}
\usepackage{upgreek}

\usepackage{amsmath}
\newcommand{\lse}{\mathsterling}

\DeclareMathOperator*\argmax{arg \, max}		% arg max
\DeclareUnicodeCharacter{2212}{-}

\usepackage{booktabs}

\newtheorem{lemma}{Lemma}

\begin{acronym}
    \acro{cnn}[CNN]{Convolutional Neural Network}
    \acro{dsne}[{\textit{d}-SNE}]{domain adaptation using Stochastic
Neighborhood Embedding}
    \acro{grl}[GRL]{gradient reversal layer}
    \acro{rdm}[RDM]{range-Doppler map}
    \acro{fada}[FADA]{Few-shot Adversarial Domain Adaptation}
    \acro{fmcw}[FMCW]{frequency-modulated continuous-wave}
    \acro{lse}[lse]{log-sum-exp}
    \acro{mdd}[MDD]{Margin Disparity Discrepancy}
    \acro{ml}[ML]{machine learning}
    \acro{da}[DA]{domain adaptation}
\end{acronym}

\def\BibTeX{{\rm B\kern-.05em{\sc i\kern-.025em b}\kern-.08em
    T\kern-.1667em\lower.7ex\hbox{E}\kern-.125emX}}
\begin{document}

\title{Unsupervised Domain Adaptation across
\acs{fmcw} Radar Configurations
Using Margin Disparity Discrepancy}

\author{\IEEEauthorblockN{Rodrigo Hernang{\'o}mez\IEEEauthorrefmark{1},
Igor Bjelakovi{\'c}\IEEEauthorrefmark{1}\IEEEauthorrefmark{3},
Lorenzo Servadei\IEEEauthorrefmark{2}\IEEEauthorrefmark{4},
and
S{\l}awomir~Sta{\'n}czak\IEEEauthorrefmark{1}\IEEEauthorrefmark{3}}
\IEEEauthorblockA{\IEEEauthorrefmark{1}
Fraunhofer Heinrich Hertz Institute, Berlin, Germany, \{firstname.lastname\}@hhi.fraunhofer.de}
\IEEEauthorblockA{\IEEEauthorrefmark{2}Infineon Technologies AG, Munich, Germany, lorenzo.servadei@infineon.com}
\IEEEauthorblockA{\IEEEauthorrefmark{3}Technical University of Berlin, Berlin, Germany}
\IEEEauthorblockA{\IEEEauthorrefmark{4}Technical University of Munich, Munich, Germany}
}

\maketitle

\begin{abstract}
Commercial
radar sensing is gaining relevance
and \acl{ml} algorithms constitute one
of the key components that are enabling
the spread of this radio technology into
areas like surveillance
or healthcare.
However, radar datasets are still scarce
and generalization cannot be yet achieved
for all radar systems, environment conditions
or design parameters. A certain degree of
fine tuning is, therefore, usually required
to deploy machine-learning-enabled radar applications.
In this work, we consider the problem of
unsupervised domain adaptation across
radar configurations in the context of deep-learning
human activity classification using
\acl{fmcw}. For that, we focus on
the theory-inspired technique of \acl{mdd},
which has already been proved successful in the area
of computer vision. Our experiments extend
this technique to radar data, achieving a comparable accuracy
to few-shot supervised approaches for the same
classification problem.
\end{abstract}

\acresetall

\begin{IEEEkeywords}
radar, machine learning, deep learning,
transfer learning, domain adaptation,
human activity classification
\end{IEEEkeywords}

\section{Introduction}
Radar, a well-established technology for several industrial
areas, has recently gained
attention for other commercial applications
like human monitoring,
presence detection or
gesture sensing~\cite{peng2017fmcw,santra2018shortrange,lien2016soli}
due to the production of
small and compact radar sensors~\cite{trotta2021soli}.
Here radar offers
some advantages in comparison with
computer vision approaches, such as
good performance under poor-lighting conditions
or privacy protection (due to the difficulty to
identify individuals from radar images).

Similarly as in the realm of computer vision,
the use of radar
often comes hand in hand with
\ac{ml} techniques,
including deep learning,
to overcome the burden of handcrafted
feature engineering~\cite{santra2020deep}.
Due to the variety
of system design parameters at hand,
such as modulation techniques or bandwidth,
these \ac{ml} algorithms are required to generalize
well under different radar setups.
This need for inter-domain generalization
is common to several \ac{ml} problems and
it has been studied in recent years under
the paradigm of
\ac{da}~\cite{redko2019advances,kouw2018introduction}.

Considered a special case of transfer learning,
\ac{da}
involves modifying
an \ac{ml} estimator that
can be trained with enough data
from a \textit{source domain},
so that its performance
increases when evaluated
with data
originated from a different \textit{target domain}.
Unlike other transfer learning approaches,
here
the mismatch between source and target domain
lies merely in a distinct probability measure
over data rather than in different input or
output spaces~\cite{kouw2018introduction}.

The reasons for \ac{da} are usually related with insufficient
or incomplete data in the target domain, which
can be overcome with the help of data from
the source domain. In
\ac{ml} classification,
the missing information
is often the labels;
this case is referred to as \textit{unsupervised} \ac{da}.
If target data are labeled, we can apply
\textit{supervised} \ac{da} techniques 
instead~\cite{redko2019advances}.

Both supervised and unsupervised \ac{da} methods have already been
investigated in the Radar-\ac{ml} community
to overcome several problems,
including individual patient differences~\cite{yin2019selfadjustable},
aspect angle variations~\cite{chen2019eliminate},
synthetic-to-real adaptation~\cite{li2020unsupervised}
or environmental differences~\cite{stephan2021radar}.
In the case of cross-configuration adaptation,
\citet{khodabakhshandeh2021domain} use supervised techniques
such as \ac{fada}~\cite{motiian2017fewshot}
or \ac{dsne}~\cite{xu2019dsne} to
adapt their trained human activity classifier to new
\ac{fmcw} radar setups using few data.

In this paper, we build on the work in \cite{khodabakhshandeh2021domain}
by applying \ac{mdd}~\cite{zhang2019bridging}.
In that way, we confirm that
this unsupervised technique, which delivers
state-of-the-art results for computer vision
datasets, also works for radar data and thus
enables cross-configuration radar-based
human activity classification based on unlabeled data.

\section{Problem Statement}
Radar-\ac{ml} classification
deals with the evaluation of a group of 
radar features $\boldsymbol{x}\in\mathcal{X}$ obtained from a target
to find the underlying class $y\in\mathcal{Y}$ that best
describes some property of the said target.
The input space
$\mathcal{X}\subset\mathds{R}^m$
is characterized by a dimension
$m$ that depends on the radar technology and preprocessing
steps, while the label space is defined as
$\mathcal{Y}=\left\lbrace 1,\ldots,k\right\rbrace$,
with $k$ being the number of classes.

In order to achieve this classification, one has first
to find a classifier $h$ that maps $\boldsymbol{x}$
into $y$.
The \ac{ml} approach assumes a sufficiently
large amount
of data available conveying information both about the inputs
and the class so that it can be used to train
$h$ among a restricted hypothesis class $\mathcal{H}$.
This dataset consists of a sequence of pairs of
features and labels, i.e.
$\left\lbrace\left(\boldsymbol{x}_i, y_i\right)\right \rbrace^n_{i=1}$,
that have been previously sampled
from a certain domain $\mathcal{D}$,
defined to be
\begin{equation}
    \mathcal{D}=\left(\mathcal{X},\mathcal{Y},p_{\mathcal{D}}\right),
\end{equation}
with an associated probability measure 
$p_{\mathcal{D}}$ over
$\mathcal{X}\times\mathcal{Y}$.
Here and hereafter, we write $\mathbf{x}$
and $\mathrm{y}$ in their upright form whenever
we refer to the random variables related to $p_{\mathcal{D}}$,
and not to its realizations.

By choosing an objective loss function
\begin{equation}
	\ell:\;\mathcal{H}\times\mathcal{X}\times\mathcal{Y}
	\rightarrow\mathds{R}_{0+}
\end{equation}
and minimizing it over the hypothesis class $\mathcal{H}$
with a suitable optimization method,
we can train an $h$ that performs well for the
available data.
The performance of $h$
can thus be measured by the \textit{risk} associated with the loss for
a domain $\mathcal{D}$. This risk $\mathcal{L}_{\mathcal{D}}$
represents the expected
value of the loss of $h$ for $p_{\mathcal{D}}$:

\begin{equation}
	\mathcal{L}_{\mathcal{D}}\left(h\right)=
	\mathds{E}_{\mathcal{D}}
	\ell\left(h, \left(\mathbf{x},\mathrm{y}\right)\right).
\end{equation}

By assuming the indicator function
$\mathds{1}_{h\left(\boldsymbol{x}\right)\neq y}$
to be the loss, we obtain the \textit{0-1 error}
$\textrm{err}_{\mathcal{D}}
    \left(h\right)
    \triangleq
    \mathds{E}_{\mathcal{D}}
    \mathds{1}_{h\left(\mathbf{x}\right)\neq \mathrm{y}}$.
In practice, we do not have access to $p_{\mathcal{D}}$,
so we resort to its empirical approximation
$\mathcal{L}_{\widehat{\mathcal{D}}}\left(h\right)
	\triangleq
	\sum_{i=1}^{n}
	{\ell\left(h, \left(\boldsymbol{x}_i,y_i\right)\right)}/n$
for a dataset $\widehat{\mathcal{D}}$
with $n$ samples drawn from $\mathcal{D}$.

If generalization is achieved,
$h$ will also behave well for
unseen data as long as it is drawn from the same domain.
Unfortunately, this assumption cannot always be guaranteed.
It is often the case that training data have been drawn from
a \textit{source domain} $\mathcal{S}$ but we would like to leverage
the trained classifier for a different \textit{target domain} $\mathcal{T}$.
Depending on how dissimilar $\mathcal{S}$ and $\mathcal{T}$ are,
the performance of the trained classifier can degrade
significantly.
In our specific problem, this \textit{domain shift}
is given by the choice of different
\ac{fmcw} settings
and presents an additional challenge
in the lack of the labels
for the training data from $\mathcal{T}$.
The absence of
labeled target data
makes it necessary to apply
\textit{unsupervised} \ac{da}.
In this paper, we explore
this possibility by using
\acf{mdd}~\cite{zhang2019bridging}.

\subsection{\acl{mdd}}

In order to use \ac{mdd},
we assume a
hypothesis class induced by a space $\mathcal{F}$
of \textit{scoring functions} $f:\:\mathcal{X}\mapsto\mathds{R}^k$.
We also introduce the shorthand $f_{y}\left(\boldsymbol{x}\right)$
to refer to the $y$-th component of $f\left(\boldsymbol{x}\right)$.
The hypothesis class is given by
\begin{equation}
    \mathcal{H}\triangleq
    \left\lbrace
    h_f:\:\boldsymbol{x}\mapsto\argmax_{y\in\mathcal{Y}}
    {f_y\left(\boldsymbol{x}\right)}\mid
    f\in\mathcal{F}\right\rbrace\,.
\end{equation}

\ac{mdd} has been developed by
\citet{zhang2019bridging}
as a practical algorithm
based on the concept of discrepancy distance by
\citet{mansour2009domain}.
For that,
they define the \textit{margin loss}
$\textrm{err}^{\left(\rho\right)}
_{{\mathcal{D}}}\left(f\right)$ as
\begin{equation}
    \textrm{err}
    ^{\left(\rho\right)}_{{\mathcal{D}}}
    \left(f\right)
    \triangleq
	\mathds{E}_{\mathcal{D}}
    {\Phi^{\left(\rho\right)}\circ
    \phi_f
    \left(\mathbf{x},
    \mathrm{y}\right)}\,,
\end{equation}\begin{equation}
    \phi_f\left(\boldsymbol{x}
    ,y\right)
    \triangleq\frac{1}{2}
    \left(f_y\left(\boldsymbol{x}\right)
    -\max_{y'\neq y}{
    f_{y'}\left(\boldsymbol{x}\right)}
    \right)\,,
\end{equation}
\begin{equation}\label{eq:margin-func}
    \Phi^{\left(\rho\right)}\left(x\right)
    \triangleq
    \begin{cases}
      0 & \rho\leq x\\
      1-x/\rho & 0\leq x\leq\rho\\
      1 & x\leq0
    \end{cases}\,,   
\end{equation}
and the true and empirical
\textit{margin disparity}
between two scoring functions $f'$ and $f$ as
\begin{equation}
    \textrm{disp}
    ^{\left(\rho\right)}_{{\mathcal{D}}}
    \left(f',f\right)
    \triangleq
	\mathds{E}_{\mathcal{D}}
    {\Phi^{\left(\rho\right)}\circ
    \phi_{f'}
    \left(\mathbf{x},
    h_f\left(\mathbf{x}\right)\right)}\,,
\end{equation}
\begin{equation}
    \textrm{disp}
    ^{\left(\rho\right)}_{\widehat{\mathcal{D}}}
    \left(f',f\right)
    \triangleq
	\frac{1}{n}
	\sum_{i=1}^{n}
    {\Phi^{\left(\rho\right)}\circ
    \phi_{f'}
    \left(\boldsymbol{x}_i,
    h_f\left(\boldsymbol{x}_i\right)\right)}\,,
\end{equation}
to finally formulate the following minimax
optimization problem:

\begin{equation}\label{eq:problem}
    \begin{gathered}
    \min_{f\in\mathcal{F}}
    {\textrm{err}
    ^{\left(\rho\right)}_{\widehat{\mathcal{S}}}
    \left(f\right)+
    d^{\left(\rho\right)}
    _{f,\mathcal{F}}
    \left(
    \widehat{\mathcal{S}},
    \widehat{\mathcal{T}}\right)
    },\\
    d^{\left(\rho\right)}
    _{f,\mathcal{F}}
    \left(
    {\mathcal{S}},
    {\mathcal{T}}\right)
    \triangleq
    \sup_{f'\in\mathcal{F}}
    {\left(\textrm{disp}
    ^{\left(\rho\right)}_{{\mathcal{T}}}
    \left(f',f\right)-
    \textrm{disp}
    ^{\left(\rho\right)}_{{\mathcal{S}}}
    \left(f',f\right)\right)}\,.
    \end{gathered}
\end{equation}
Following the principles
of unsupervised
\ac{da},
the  \ac{mdd} term
$d^{\left(\rho\right)}
    _{f,\mathcal{F}}$
does not make use of
any labels $y_i$.
Furthermore, the solution to \eqref{eq:problem}
minimizes the 0-1 error of
$h_f$ in the target domain,
as
the authors of~\cite{zhang2019bridging}
prove with the following theoretical bound:

\begin{equation}\label{eq:bound}
    \textrm{err}_{\mathcal{T}}
    \left(h_f\right)\leq
    \textrm{err}^{\left(\rho\right)}_{\mathcal{S}}
    \left(f\right)+
    d^{\left(\rho\right)}
    _{f,\mathcal{F}}
    \left(
    {\mathcal{S}},
    {\mathcal{T}}\right)+\lambda\,,
\end{equation}
where $\lambda$ is the ideal combined
margin loss:

\begin{equation}
    \lambda=
    \min_{f^{*}\in\mathcal{F}}
    {\left\lbrace
    \textrm{err}^{\left(\rho\right)}_{\mathcal{S}}
    \left(f^{*}\right)+
    \textrm{err}^{\left(\rho\right)}_{\mathcal{T}}
    \left(f^{*}\right)
    \right\rbrace}\,.
\end{equation}

The bound in \eqref{eq:bound} can also be expressed in terms
of empirical measures rather than true probability measures
by the addition of
Rademacher complexity terms~\cite[Chapter 3]{mohri2018foundations}.

Despite the interesting properties
of \ac{mdd},
$\Phi^{\left(\rho\right)}\circ\phi_f$ is non-smooth, non-convex
and its training causes vanishing and exploding gradients, which leads \citet{zhang2019bridging}
to work with
the cross-entropy loss instead.
For this, they map
$f\left(\boldsymbol{x}\right)$ to the
$k$-simplex via the softmax function
$\sigma$, as it is customary in deep learning,
where the elements of $\sigma\left(\boldsymbol{z}\right)$
are given by
\begin{equation}
    \sigma_j\left(\boldsymbol{z}\right)
    \triangleq
    \frac{\exp{z_j}}
    {\sum_{i=1}^k
    {\exp{z_i}}},\;\;\text{for }
    j=1,\ldots,k\,.
\end{equation}
The composition
of the cross-entropy loss with the softmax yields the
\ac{lse}
loss $\lse_f$:

\begin{multline}\label{eq:lse}
    \lse_f\left(\boldsymbol{x},y
    \right)\triangleq
    H\left(\mathds{1}_{j=y},
    \sigma\left(
    f\left(\boldsymbol{x}\right)\right)\right)=
    -\log{\sigma_y\left(
    f\left(\boldsymbol{x}\right)\right)}\\
    =\log{\sum_{y'\in\mathcal{Y}}
    {\exp{\left(f_{y'}\left(\boldsymbol{x}\right)-
        f_{y}\left(\boldsymbol{x}\right)\right)}}}\,.
\end{multline}

\citet{zhang2019bridging} propose to use
$\lse_f$ instead of
$\Phi^{\left(\rho\right)}\circ\phi_f$
for
$\textrm{err}^{\left(\rho\right)}_{\widehat{\mathcal{S}}}$
and
$\textrm{disp}^{\left(\rho\right)}_{\widehat{\mathcal{S}}}$
in \eqref{eq:problem}.
As for
$\textrm{disp}^{\left(\rho\right)}_{\widehat{\mathcal{T}}}$,
they use the adversarial loss
$\widetilde{\lse}_{f}$ proposed by
\citet{goodfellow2014generative}, i.e.
\begin{equation}
   \widetilde{\lse}_{f}\left(\boldsymbol{x},y
    \right)\triangleq
    \log{\left(1-\sigma_y\left(
    f\left(\boldsymbol{x}\right)\right)\right)}\,,
\end{equation}
so that their \ac{mdd} ultimately becomes
\begin{multline}\label{eq:mdd_pract}
\widetilde{d}^{\left(\gamma\right)}
    _{f,\psi,\mathcal{F}}
    \left(
    {\widehat{\mathcal{S}}},
    {\widehat{\mathcal{T}}}\right)
    \triangleq%\\
    \max_{f'\in\mathcal{F}}
    \mathds{E}_{\boldsymbol{x}^t
    \sim\widehat{\mathcal{T}}}{
      \widetilde{\lse}_{f'}\left(\psi\left(\boldsymbol{x}^t\right),
      h_f\left(\psi\left(\boldsymbol{x}^t\right)\right)\right)}\\-
      \gamma\mathds{E}_{\boldsymbol{x}^s
    \sim\widehat{\mathcal{S}}}{
      \lse_{f'}\left(\psi\left(\boldsymbol{x}^s\right),
      h_f\left(\psi\left(\boldsymbol{x}^s\right)\right)\right)}
\end{multline}
for a \textit{margin factor} $\gamma>0$
and a feature extractor $\psi$ that levels
the min-player to the max-player~\cite{zhang2019bridging}
(A concrete example of $\psi$ is given
in \eqref{eq:feat_ext}, \Cref{sec:experiment}).
The authors explain
that this is equivalent to
the use of the margin loss with a margin
$\rho=\log{\gamma}$
and that the problem is still
solved for $\mathcal{S}=\mathcal{T}$~\cite{zhang2019bridging}.

In addition to the results
in~\cite{zhang2019bridging},
we observe that the use
of the recently proposed soft-margin softmax
\cite{liang2017softmargin} instead of
$\sigma$
in \eqref{eq:lse} provides
an upper bound for
$\textrm{err}^{\left(\rho\right)}_{\mathcal{S}}$.
The entries of the soft-margin softmax
$\sigma^{\left(\rho\right)}\left(\boldsymbol{z}\right)$
are defined as
\begin{multline}
    \sigma^{\left(\rho\right)}_j\left(\boldsymbol{z}\right)
    \triangleq
    \frac{\exp{\left(z_j-\rho\right)}}
    {\exp{\left(z_j-\rho\right)}+\sum_{i\neq j}
    {\exp{z_i}}},\\
    \text{for }j=1,\ldots,k;\;\;\rho\in\mathds{R}_{+}
\end{multline}
and this induces the soft-margin cross-entropy loss $\lse_f^{\left(\rho\right)}$:
\begin{multline}
	\lse_f^{\left(\rho\right)}
	\left(\boldsymbol{x},y\right)
	\triangleq -\log{\sigma^{\left(\rho\right)}_y\left(
    f\left(\boldsymbol{x}\right)\right)}\\
    =
	\log{\sum_{y'\in\mathcal{Y}}
    {\exp{\left(f_{y'}\left(\boldsymbol{x}\right)-
        f_{y}\left(\boldsymbol{x}\right)
        +\rho\cdot\mathds{1}_{y'\neq y}\right)}}}\,.
\end{multline}
Likewise, a soft-max adversarial loss can also
be defined as
\begin{equation}
\widetilde{\lse}^{\left(\rho\right)}_{f}
\left(\boldsymbol{x},y
\right)\triangleq
\log{\left(1-\sigma^{\left(\rho\right)}_y\left(
f\left(\boldsymbol{x}\right)\right)\right)}\,.
\end{equation}

We prove this soft-margin-based
bound with the help of the following Lemma,
which motivates us to
investigate
$\lse^{\left(\rho\right)}_f$
further in \cref{sec:experiment}.

\begin{lemma}\label{th:lse}
The soft-max cross entropy bounds the margin loss as
\begin{equation}\label{eq:lemma}
    \Phi^{\left(\rho\right)}\circ     \phi_f
    \left(\boldsymbol{x},y\right)\leq
    \frac{1}{2\rho}
    \lse^{\left(2\rho\right)}_f
    \left(\boldsymbol{x},y\right)\,.
\end{equation}
\end{lemma}
\begin{proof}
First, let us recall the generalized hinge 
loss~\cite{shalev-shwartz2014understanding}:
%\citep[Section 17.2]{shalev-shwartz2014understanding}:
\begin{equation}
    \hslash_f^{\left(\theta\right)}
    \left(\boldsymbol{x},y\right)
    \triangleq
    \max_{y'\in\mathcal{Y}}
    {\left(f_{y'}\left(\boldsymbol{x}\right)-
        f_{y}\left(\boldsymbol{x}\right)
        +\theta\cdot\mathds{1}_{y'\neq y}\right)}\,.
\end{equation}
Noting that the term within $\max_{y'\in\mathcal{Y}}$
is null $\forall y'=y$, we have
\begin{multline}
\label{eq:hinge}
    \hslash_f^{\left(2\rho\right)}
    \left(\boldsymbol{x},y\right)=
    \max\left\lbrace 0,
    \max_{y'\neq y}{
    f_{y'}\left(\boldsymbol{x}\right)
    -f_y\left(\boldsymbol{x}\right)
    +2\rho}
    \right\rbrace=\\
    \max\left\lbrace 0,
    2\rho- 2\phi_f\left(\boldsymbol{x},y\right)
    \right\rbrace=%\\=
    2\rho\max\left\lbrace 0,
    1-\phi_f\left(\boldsymbol{x},y\right)/\rho
    \right\rbrace\,.
\end{multline}

The last expression in \eqref{eq:hinge}
can be derived from \eqref{eq:margin-func}
if we set $1-x/\rho$ instead of 1
for $x\leq0$, hence
\begin{equation}\label{eq:hinge-ineq}
    \Phi^{\left(\rho\right)}\circ\phi_f
\left(\boldsymbol{x},y\right)\leq
\frac{1}{2\rho}
\hslash^{\left(2\rho\right)}_f
\left(\boldsymbol{x},y\right)
\end{equation}
and the proof is concluded using the fact that
\begin{equation}\label{eq:lse-ineq}
    \log{\sum_{a\in\mathcal{A}}{\exp{a}}}
\geq\max_{a\in\mathcal{A}}{{a}}
\end{equation}
for any finite set
$\mathcal{A}$.
\end{proof}

Taking the expectation w.r.t.
$p_{\mathcal{S}}$ in \eqref{eq:lemma},
we finally obtain
\begin{equation}\label{eq:loss-bound}
    \textrm{err}^{\left(\rho\right)}_{\mathcal{S}}
    \left(f\right)
    \leq\frac{1}{2\rho}
    \mathds{E}_{\mathcal{S}}
    {\lse^{\left(2\rho\right)}_f
    \left(\mathbf{x},\mathrm{y}\right)}\,.
\end{equation}
Despite the gap introduced by \cref{eq:hinge-ineq,eq:lse-ineq}, we note that \cref{eq:loss-bound} delivers convincing bounds
for a small $\textrm{err}^{\left(\rho\right)}_{\mathcal{S}}
\left(f\right)$,
which can be achieved by training
$f$ under enough samples from $\mathcal{S}$.

\section{Experiments}
\label{sec:experiment}

\subsection{Setup}\label{sec:setup}

Similarly as in \cite{khodabakhshandeh2021domain}, we focus on human activity recognition
using data that have been measured
simultaneously with 4 different
60-GHz \ac{fmcw} radar sensors.
For these measurements,
2 male subjects were recorded separately while
performing 5 different activities:
\textit{standing}, \textit{waving},
\textit{walking}, \textit{boxing}
or \textit{boxing while walking}.
Each one of the radar sensors was configured with a different set of radar parameters,
presented in \cref{table:configs} as \textbf{I} to \textbf{IV}.
Here, the divergent parameters of the different configurations (marked in bold)
affect the temporal and range resolution of the
\ac{rdm} sequences,
as well as the maximum observable scope of the latter.
From these configurations, \textbf{I} has been taken over from~\cite{khodabakhshandeh2021domain}.

\begin{table}[!t]    % [htp]
    \caption{Radar configuration parameters}
    \begin{center}
    % options A, B, C, D
\begin{tabular}{llrrrr}
    \multicolumn{2}{l}{Configuration name}  & \textbf{I} & \textbf{II}
                            & \textbf{III} & \textbf{IV}  \\ \hline
    Chirps per frame &    $n_c$   & 64       & 64       & 64       & 64        \\
    Samples per chirp &  $n_s$    & 256      & 256      & \textbf{128}      & 256        \\
    Bandwidth & [\si{\giga\Hz}]                   & 2       & \textbf{1}        & 2        & 2          \\
    Frame period & [\si{\milli\second}]      & \textbf{50}       & 32       & 32      &     32     \\
    Chirp to chirp time & [\si{\micro\second}]    & 250     & 250   & 250   & 250   \\
    Range resolution & [\si{\centi\meter}]  & 7.5      & \textbf{15}     & 7.5     & 7.5       \\
    Max. range & [\si{\meter}]         & 6.2   & \textbf{12.5}   & \textbf{4.8}   & 6.2     \\
    Max. speed & [\si{\meter\per\second}]        & 5.0     & 5.0     & 5.0     & 5.0       \\
    Speed resolution & [\si{\meter\per\second}] & 0.15     & 0.15     & 0.15     & 0.15      
\end{tabular}
    \label{table:configs}
    \end{center}
\end{table}

\begin{figure}[!b] %[hb]
    \centering
    \def\svgwidth{0.5\textwidth}
    \footnotesize
    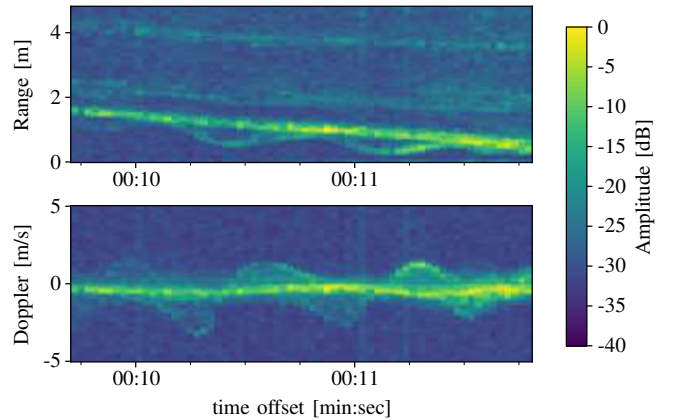
    \caption{Range and Doppler spectrogram
    for \textit{boxing while walking}.}
    \label{fig:spectrograms}
\end{figure}

The input features $\boldsymbol{x}$ 
in~\cref{fig:spectrograms}
comprise both
range ($\boldsymbol{x}_r$) and
Doppler ($\boldsymbol{x}_d$) information, i.e.:
\begin{equation}\label{eq:x}
    \boldsymbol{x}=\left(\boldsymbol{x}_r, \boldsymbol{x}_d\right),\;\;\boldsymbol{x}_r, \boldsymbol{x}_d\in\mathds{R}^{64\times128}.
\end{equation}
The radar preprocessing to produce $\boldsymbol{x}$ is also based
on~\cite{khodabakhshandeh2021domain}, with the notable addition of cropping and resampling of the
spectrograms to ensure
the dimensions in~\eqref{eq:x}
and the scopes of
\SIrange{0}{2}{\second},
\SIrange{0.0}{4.8}{\meter} and
\SIrange{-5}{5}{\meter\per\second}
for the time, range and Doppler
dimensions, respectively.

Despite all the preprocessing,
the differences on resolution
still yield a domain shift across configurations that we
try to tackle with \ac{mdd}.
For that, we take both spectrograms as an input to our feature extractor
$\psi$. Here we choose the same topology as in~\cite{khodabakhshandeh2021domain};
that is, a pair of twin branches $\psi_r$ and $\psi_d$, each one
consisting of 3 convolutional layers for which we concatenate the outputs:
\begin{equation}\label{eq:feat_ext}
    \psi\left(\boldsymbol{x}\right)\equiv
    \left(\psi_r\left(\boldsymbol{x}_r\right),
        \psi_d\left(\boldsymbol{x}_d\right)\right).
\end{equation}

Furthermore, we employ a bottleneck layer of 512 nodes and choose our
hypothesis space $\mathcal{F}$ to be consistent
with the structure of the fully connected layers from~\cite{khodabakhshandeh2021domain}.

Motivated by \cref{th:lse},
we replace the vanilla loss terms in
\eqref{eq:mdd_pract}
by the soft-margin losses
$\lse^{\left(\rho\right)}$ and
$\widetilde{\lse}^{\left(\rho\right)}$
with
\begin{equation}
    \rho=2\log{2}\simeq1.386
\end{equation}
and we set the margin factor $\gamma=1$
since the margin
$\rho$ is already included in the loss,
effectively rendering \ac{mdd} as
\begin{multline}\label{eq:mdd_sm}
\widehat{d}^{\left(\rho\right)}
    _{f,\psi,\mathcal{F}}
    \left(
    {\widehat{\mathcal{S}}},
    {\widehat{\mathcal{T}}}\right)
    \triangleq
    \max_{f'\in\mathcal{F}}
    \mathds{E}_{\boldsymbol{x}^t
    \sim\widehat{\mathcal{T}}}{
      \widetilde{\lse}^{\left(\rho\right)}_{f'}
      \left(\psi\left(\boldsymbol{x}^t\right),
      h_f\left(\psi\left(\boldsymbol{x}^t\right)\right)\right)}\\-
      \mathds{E}_{\boldsymbol{x}^s
    \sim\widehat{\mathcal{S}}}{
      \lse^{\left(\rho\right)}_{f'}
      \left(\psi\left(\boldsymbol{x}^s\right),
      h_f\left(\psi\left(\boldsymbol{x}^s\right)\right)\right)}\,.
\end{multline}

Other than that, we leave all hyperparameters
to the same values as in~\cite{zhang2019bridging}
and adapt their implementation as
in~\cref{fig:mdd}.
This has been written in Pytorch as an instance of
adversarial training, where
a \ac{grl} is used to minimize
the \ac{mdd} loss term on
$\psi$
while maximizing on $f'$ as
the minimax formulation in \eqref{eq:problem}
mandates.

\begin{figure}[!b] %[hb]
    \centering
    \def\svgwidth{0.5\textwidth}
    \footnotesize
    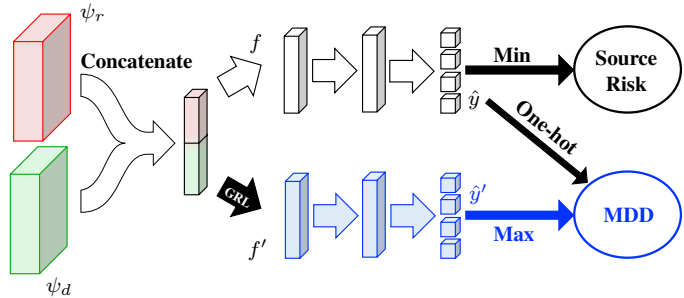
    \caption{\ac{mdd} adversarial network,
    adapted from~\cite{zhang2019bridging}.}
    \label{fig:mdd}
\end{figure}

The number of samples per dataset lies over 1150 samples for the train sets
and over 350 samples for the test sets.

\subsection{Results}

We have run unsupervised training experiments for all possible domain pairs within
configurations \textbf{I}-\textbf{IV} and summarized the resulting test accuracies
on the test sets
in \cref{table:acc}.

The figures follow the same trend as the
results of \ac{mdd} in
computer vision datasets as
reported by \citet{zhang2019bridging}.
and presented in~\cref{table:mdd}.
Here one can compare, for instance,
the results using \ac{mdd}
for our \ac{fmcw} data with the minimum
and maximum accuracies obtained
for Office-31,
a dataset containing
4,652 images from three domains~\cite{saenko2010adapting}.
It is also noteworthy that the highest accuracy
for \ac{fmcw} exceeds both that of
the Office-Home dataset
(15,500 images from four domains)~\cite{venkateswara2017deep}
and the VisDA dataset
(280K real and synthetic images)~\cite{peng2017visda}.

Our results are also comparable with the
\ac{fada} method for \ac{fmcw}-based
human activity recognition in~\cite{khodabakhshandeh2021domain},
which increases the baseline accuracy
of
\SIrange[range-phrase=--]{50}{60}{\percent}
without \ac{da} to
\SIrange[range-phrase=--]{88}{92}{\percent}.
Here it is important to note
that \ac{mdd} is,
in contrast to \ac{fada},
an unsupervised technique and
thus it presents the advantage
of working with unlabeled target data.

\begin{table}[!t]    % [hp]
\caption{Test accuracy [\si{\percent}] of \ac{mdd} for
\ac{fmcw} data}
\begin{center}
\begin{tabular}{lrcccc}
 & \multicolumn{5}{c}{Target configuration} \\
 \multirow{5}{*}{\rotatebox[origin=c]{90}{
\begin{tabular}[c]{@{}c@{}}Source\\ configuration\end{tabular}}}
 &     & I    & II   & III  & IV   \\
 & I   & -    & 91.4 & 90.6 & 88.3 \\
 & II  & 90.9 & -    & 89.8 & 89.4 \\
 & III & 89.4 & 90.4 & -    & 89.4 \\
 & IV  & 92.5 & 85.8 & 90.9 & -   
\end{tabular}
\label{table:acc}
\end{center}
\end{table}

\begin{table}[!t]    % [hp]
\caption{Min. and max. accuracy [\si{\percent}]
of \ac{mdd} for different datasets}
\begin{center}
\begin{tabular}{cccc}
Office-31  & Office-Home & VisDa & FMCW      \\
72.2-100.0 & 53.6-82.3   & 74.6 (single value)  & 85.8-92.5
\end{tabular}
\label{table:mdd}
\end{center}
\end{table}

\begin{table}[!t]    % [hp]
\caption{Average accuracy comparison [\si{\percent}] of the original
\ac{mdd} implementation and the soft-margin version}
\begin{center}
\begin{tabular}{llll}
                & Office-31 & Office-Home & \ac{fmcw}  \\
Original \ac{mdd}    & 88.9      & 68.1        & 89.525 \\
Soft-margin \ac{mdd} & 88.3      & 67.6        & 89.9  
\end{tabular}
\label{table:sm}
\end{center}
\end{table}

We have also compared the
average accuracy across domain combinations
for the
original implementation
of \ac{mdd} in \eqref{eq:mdd_pract}
with the average accuracy for
our soft-margin version in \eqref{eq:mdd_sm},
taking the Office-31 and Office-Home datasets
as well as our \ac{fmcw} radar
data. The results, which
can be seen in \cref{table:sm}, show little
difference between both implementations.

\section{Conclusion}
In this work, we confirm that
the \ac{mdd} algorithm,
which has already shown promising results
for unsupervised \ac{da} in
the area of computer vision, is also suitable
for radar data across different \ac{fmcw} parameters.
The obtained accuracy can become as high
as for some supervised techniques~\cite{khodabakhshandeh2021domain}
while using a much more limited dataset, paving thus the way for a prompt deployment of
radar-based deep learning applications with
custom configurations.

In our experiments,
we observe that
the use of the soft-margin cross entropy loss provides similar results as the original implementation by
\citet{zhang2019bridging}.
Since the motivation of \ac{mdd} is
to bring the algorithms closer to the analytical
performance bounds of \ac{da}, we see
potential in this alternative loss function
to bridge the gap between theory and practice.

\section*{Acknowledgment}

We gratefully acknowledge the support of NVIDIA Corporation with the donation of the DGX-1
used for this research. We would also like to thank
Avik Santra from Infineon Technologies AG
for his support throughout this work.

\bibliographystyle{IEEEtranN}
\bibliography{mdd_radar}

\end{document}